\documentclass[11pt]{article}

\usepackage[toc,page]{appendix}
\usepackage[normalem]{ulem} 
\usepackage{amssymb} 
\usepackage{amsmath} 
\usepackage{amsthm} 
\usepackage{setspace} 
\usepackage{cite} 
\usepackage{hyperref}

\usepackage{tikz}
\usetikzlibrary{calc,positioning,shapes,shadows,arrows,fit}

\usepackage[hmargin=2.5cm,vmargin=3cm]{geometry}

\newtheorem{definition}{Definition} 
\newtheorem{corollary}{Corollary} 
 
\newtheorem{claim}{Claim} 
\newtheorem{lemma}{Lemma} 
\newtheorem{theorem}{Theorem}

\newtheorem{proposition}{Proposition}
\newtheorem{remark}{Remark}

\newcommand{\F}{\mathbb{F}}

\newcommand{\M}{\mathcal{M}}

\renewcommand{\angle}[1]{\mathopen{\langle} #1\mathclose{\rangle}}

\newcommand{\sgn}{\mbox{\small\rm sgn}}

\DeclareMathOperator{\rank}{\mbox{\small\rm rank}}

\DeclareMathOperator{\poly}{\mbox{\small\rm poly}}

\title{Identity Testing for +-Regular Noncommutative Arithmetic Circuits}

\author{V. Arvind\thanks{Institute of Mathematical Sciences (HBNI), Chennai,
    India, \texttt{email: arvind@imsc.res.in}} \and Pushkar S
  Joglekar\thanks{Vishwakarma Institute of Technology, Pune, India,
    \texttt{email: joglekar.pushkar@gmail.com}}  \and Partha
  Mukhopadhyay\thanks{Chennai Mathematical Institute, Chennai, India,
    \texttt{email: partham@cmi.ac.in}} \and S. Raja\thanks{Chennai Mathematical Institute, Chennai, India,
    \texttt{email: sraja@cmi.ac.in}}}
 
\begin{document} 

\maketitle

\begin{abstract} 
An efficient randomized polynomial identity test for noncommutative
polynomials given by noncommutative arithmetic circuits remains an
open problem. The main bottleneck to applying known techniques is that
a noncommutative circuit of size $s$ can compute a polynomial of
degree exponential in $s$ with a double-exponential number of nonzero
monomials. In this paper, which is a follow-up on our earlier article
\cite{AMR16}, we report some progress by dealing with two natural
subcases (both allow for polynomials of exponential degree and a
double exponential number of monomials):

\begin{enumerate}
 \item We consider \emph{$+$-regular} noncommutative circuits: these
   are homogeneous noncommutative circuits with the additional
   property that all the $+$-gates are layered, and in each $+$-layer
   all gates have the same syntactic degree. We give a
   \emph{white-box} polynomial-time deterministic polynomial identity
   test for such circuits. Our algorithm combines some new structural
   results for $+$-regular circuits with known results for
   noncommutative ABP identity testing \cite{raz05PIT}, rank bound of
   commutative depth three identities \cite{SS13}, and equivalence
   testing problem for words \cite{Loh15, MSU97, Pla94}.
  
\item Next, we consider $\Sigma\Pi^*\Sigma$ noncommutative circuits:
  these are noncommutative circuits with layered $+$-gates such that
  there are only two layers of $+$-gates. These $+$-layers are the
  output $+$-gate and linear forms at the bottom layer; between the
  $+$-layers the circuit could have any number of $\times$ gates.  We
  given an efficient randomized \emph{black-box} identity testing
  problem for $\Sigma\Pi^*\Sigma$ circuits. In particular, we show if
  $f\in\F\angle Z$ is a nonzero noncommutative polynomial computed by
  a $\Sigma\Pi^*\Sigma$ circuit of size $s$, then $f$ cannot be a
  polynomial identity for the matrix algebra $\mathbb{M}_s(F)$, where
  the field $F$ is a sufficiently large extension of $\F$ depending on
  the degree of $f$.
\end{enumerate}
\end{abstract}

\section{Introduction}

Noncommutative arithmetic computation is an important area of
algebraic complexity theory, introduced by Hyafil \cite{Hya77} and
Nisan \cite{N91}. The main algebraic structure of interest is the free
noncommutative ring $\F\angle{Z}$ over a field $\F$, where
$Z=\{z_1,z_2,\cdots,z_n\}$ is a set of $n$ free noncommuting
variables. The elements of $\F\angle{Z}$ are \emph{noncommutative
  polynomials} which are $\F$-linear combinations of \emph{monomials},
which, in turn, are words in $Z^*$.

An important algorithmic problem in this area is noncommutative
polynomial identity testing: The input is a noncommutative polynomial
$f\in \F \angle{Z}$ computed by a noncommutative arithmetic circuit
$C$.  The polynomial $f$ can be either given by black-box access to
$C$ (using which we can evaluate $C$ on matrices with entries from
$\F$ or an extension field), or the circuit $C$ may be explicitly
given as the input. The algorithmic problem is to check if the
polynomial computed by $C$ is identically zero.

We recall the formal definition of a noncommutative arithmetic
circuit.

\begin{definition}
A \emph{noncommutative arithmetic circuit} $C$ over a field $\F$ and
indeterminates $z_1,z_2,\cdots,z_n$ is a directed acyclic graph (DAG)
with each node of indegree zero labeled by a variable or a scalar
constant from $\F$: the indegree $0$ nodes are the input nodes of the
circuit. Each internal node of the DAG is of indegree two and is
labeled by either a $+$ or a $\times$ (indicating that it is a plus
gate or multiply gate, respectively). Furthermore, the two inputs to
each $\times$ gate are designated as left and right inputs which is
the order in which the gate multiplication is done. A gate of $C$ is
designated as \emph{output}.  Each internal gate computes a polynomial
(by adding or multiplying its input polynomials), where the polynomial
computed at an input node is just its label. The \emph{polynomial
  computed} by the circuit is the polynomial computed at its output
gate. An arithmetic circuit is a formula if the fan-out of every gate
is at most one.
\end{definition}

Notice that if the size of circuit $C$ is $s$ the degree of the
polynomial computed by $C$ can be $2^s$. In the earlier
result \footnote{ We also note here that Raz and Shpilka
  \cite{raz05PIT} give a white-box deterministic polynomial-time
  identity test for noncommutative algebraic branching programs
  (ABPs). The results of Forbes-Shpilka \cite{FS13} and Agrawal et al.,
  \cite{AGKS15} give (among other results) a quasi-polynomial time black-box algorithm for
  small degree noncommutative ABPs.} by Bogdanov and Wee \cite{BW05}, a randomized
polynomial-time algorithm was shown for the case when the degree of
the circuit $C$ is polynomially bounded in $s$ and $n$
\cite{BW05}. The idea of the algorithm is based on a classical result
of Amitsur-Levitzki \cite{AL}.  We recall the
Amitsur-Levitzki theorem.

\begin{theorem}[Amitsur-Levitzki Theorem]\label{thm_al}
For any field $\F$, a nonzero
noncommutative polynomial $P \in \F\angle{Z}$ of degree $\leq 2d-1$ cannot
be a polynomial identity for the matrix algebra $\mathbb{M}_d(\F)$.
I.e.\ $P$ does not vanish on all $d\times d$ matrices over $\F$.
\end{theorem}

Bogdanov and Wee's randomized PIT algorithm \cite{BW05} applies the
above theorem to obtain a randomized PIT as follows: Let
$C(z_1,z_2,\cdots,z_n)$ be a circuit of syntactic degree bounded by
$2d-1$. For each $i\in [n]$, substitute the variable $z_i$ by a
$d\times d$ matrix $M_i$ of commuting indeterminates.  More precisely,
the $(\ell,k)^{th}$ entry of $M_i$ is $z^{(i)}_{\ell,k}$ where $1\leq
\ell,k\leq d$. By Theorem \ref{thm_al}, the matrix $M_f = f(M_1, M_2,
\ldots, M_n)$ is not identically zero. Hence, in $M_f$ there is an
entry $(\ell',k')$ which has the \emph{commutative} nonzero polynomial
$g_{\ell',k'}$ over the variables $\{z^{(i)}_{\ell,k} : 1\leq i\leq n,
1\leq \ell,k\leq d\}$. Notice that the degree of the polynomial
$g_{\ell',k'}$ is at most $2d-1$. If we choose an extension field of
$\F$ of size at least $4d$, then we get a randomized polynomial
identity testing algorithm by the standard
Schwartz-Zippel-Lipton-DeMello Lemma \cite{Sch80,Zippel79,DL78}. 


The problem with this approach for general noncommutative circuits
(whose degree can be $2^s$) is that the dimension of the matrices
grows linearly with the degree of the polynomial. Therefore, this
approach only yields a randomized exponential-time algorithm for the
problem. Finding an efficient randomized identity test for general
noncommutative circuits is a well-known open problem, as mentioned in
a recent workshop on algebraic complexity theory\cite{WACT16}.

In \cite{AMR16}, we made partial progress on this problem: we gave an
efficient randomized black-box polynomial identity test for
noncommutative arithmetic circuits that compute a polynomial with
\emph{exponentially many} monomials \cite{AMR16}. However, in general
noncommutative circuits of size $O(s)$ can compute polynomials with
$2^{2^{s}}$ monomials. For example the polynomial $f(x,y)=(x+y)^{2^s}$
has noncommutative circuit of size $O(s)$ but the number of monomials
is $2^{2^s}$.

\section{Main Results} 

We first consider identity testing for a subclass of homogeneous
noncommutative circuits, that we call $+$-regular circuits. These are
syntactic homogeneous circuits where the $+$-gates can be partitioned
into layers such that: (i) there is no directed paths between the
$+$-gates in a same layer, and (ii) all $+$-gates in a layer have the
same syntactic degree.  The output gate is a $+$ gate. We give a
deterministic white box polynomial-time algorithm that tests whether a
given noncommutative $+$-regular circuit of size $s$ computes the
identically zero polynomial.

\begin{theorem}\label{regular-thm}
Let $C$ be a noncommutative $+$-regular circuit of size $s$ given as a
white-box computing a polynomial in $\F\angle X$. There is a
deterministic polynomial-time algorithm that tests whether $C$
computes the identically zero polynomial.
\end{theorem}

Next, we consider $\Sigma\Pi^*\Sigma$ noncommutative circuits. These
are noncommutative circuits with layered $+$-gates such that there are
only two layers of $+$-gates. These $+$-layers are the output $+$-gate
and linear forms at the bottom layer; between the $+$-layers the
circuit could have any number of $\times$ gates. We give an efficient
randomized black-box polynomial identity test for $\Sigma\Pi^*\Sigma$
circuits. More precisely, we show the following result. 

\begin{theorem}\label{depth-3-pit}
Let $\F$ be a field of size more than $D$. Let
$f(x_1,\ldots,x_n)\in\F\angle X$ be a nonzero polynomial of degree $D$
computed by a homogeneous $\Sigma\Pi^{*}\Sigma$ circuit with top gate
fan-in $s$ and the fan-in of the product gates bounded by $D$. Then
$f$ cannot be a polynomial identity for the matrix algebra
$\mathbb{M}_s(\F)$.
\end{theorem}

\begin{corollary}\label{depth-3-algo}
Let $C$ be a homogeneous $\Sigma\Pi^*\Sigma$ circuit of size $s$
computing a polynomial $f(x_1,\ldots,x_n)\in \F\angle X$, where $C$ is
given by black-box access. There is a randomized $\poly(s,n)$ time
algorithm that checks whether $f$ is identically zero.
\end{corollary}

\subsection*{Outline of the proofs}

We give an informal outline of the proofs for Theorem
\ref{regular-thm}, and Theorem \ref{depth-3-pit}.

\subsection*{White-box algorithm for $+$-regular circuits}

Now we informally describe the proof of Theorem \ref{regular-thm}. We
note a crucial observation: Let $T(z_1,\ldots,z_s)$ be a homogeneous
noncommutative polynomial of degree $d$. Let $R_1, \ldots,R_s$ be
homogeneous noncommutative polynomials each of degree $d'$. Consider
any maximal $\F$-linearly independent subset of the polynomials $R_1,
\ldots, R_s$.  Let $R_1, \ldots, R_k$ be such a set. We can express
$R_j = \sum_{i=1}^k \alpha_{ji} R_i$ for $k+1\leq j\leq s$ where
$\alpha_{ji}\in \F$. Then it turns out that $T(R_1, \ldots, R_k,
\sum_{i=1}^k \alpha_{k+1 i} R_i, \ldots, \sum_{i=1}^k \alpha_{s i}
R_i)=0$ if and only if $T(y_1, \ldots, y_k, \sum_{i=1}^k \alpha_{k+1
  i} y_i, \ldots, \sum_{i=1}^k \alpha_{s i} y_i)=0$ where
$y_1,\ldots,y_k$ are fresh noncommuting variables. As a consequence,
it turns out that for a deterministic polynomial-time white-box
identity testing for $+$-regular circuits, it suffices to solve the
following computational problem:

Let $P_1,\ldots, P_{\ell} \in\F\angle X$ be products of homogeneous
linear forms given by \emph{multiplicative circuits} of size $s$. The
degrees of the polynomials $P_i$ could be exponential in $s$. Then
find a maximal $\F$-linearly independent subset of the polynomials and
express the others as linear combination of the independent
polynomials.  We solve the above problem in deterministic polynomial
time. We prove that it suffices to replace $P_i$ with $\tilde{P}_i$
which is obtained from $P_i$ by retaining, in the product, only linear
forms that appear in at most $\ell^5$ locations (roughly).  This is shown using
a rank bound result of commutative depth three identities \cite{SS13}.
We also require algorithms \cite{Loh15, Pla94, MSU97} over words to
efficiently find the linear forms appearing in those $\ell^5$
locations.  Since $\tilde{P}_i : 1\leq i\leq \ell$ are small degree,
we are in the usual regime of low-degree noncommutative polynomials,
and can adapt the noncommutative ABP identity testing \cite{raz05PIT}
to solve the linear independence testing problem.

\subsection*{Black-box algorithm for homogeneous $\Sigma\Pi^*\Sigma$ circuits}

Now, we briefly sketch the proof of Theorem \ref{depth-3-pit}. Suppose
$P_1, P_2,\ldots,P_s$ are $D$-products of homogeneous linear forms in
$\F\angle X$. Consider any $\F$-linear combination $\sum_{i=1}^s
\beta_i P_i$ where, w.l.o.g $\forall i:
\beta_i\in\F\setminus\{0\}$. Then there is a subset of indices
$I\subseteq [D]$ with $|I|\leq s-1$ with the following property: For
each $i$, let $P_{i,I}$ be the polynomial obtained from $P_i$ by
treating only the variables appearing in positions in $I$ as
noncommutative, and variables in all other positions as
commutative. Then
\[
\sum_{i=1}^s \beta_i P_i = 0 \textrm{ iff } \sum_{i=1}^s \beta_i P_{i,I} = 0.
\]
Now, we can design small nondeterministic substitution automata that
can nondeterministically effect this transformation from $P_i$ to
$P_{i,I}$ for each $i$. guess the locations in $I$. The rest of the
proof is similar to the proof of Theorem~2 in \cite{AMR16}.


\section{Preliminaries}\label{prelim}

We state some useful properties of noncommutative polynomials.

\begin{proposition}\label{invert}
Let $A: \F^n\rightarrow\F^n$ be any invertible linear transformation,
and $f(x_1,x_2,\ldots,x_n)\in\F\angle X$ be any homogeneous polynomial
of degree $d$. Let $A_j(f)$ be the polynomial obtained by replacing
the variables $x_i$ appearing in the position $j\in[d]$ by
$A(x_i)$. Then $f(x_1,\ldots,x_n)\neq 0$ if and only if
$A_jf(x_1,\ldots,x_n)\neq 0$.
\end{proposition}

\begin{proof}
If $f(x_1,\ldots,x_n)=0$ then clearly
$A_jf(x_1,\ldots,x_n)=0$. Suppose $f(x_1,\ldots,x_n)\neq 0$. Write $f
= \sum_{(m_1,m_2)} f_{m_1,m_2}$, where
\[
f_{m_1,m_2}= m_1 ~L_{m_1,m_2}~ m_2,
\]
and $(m_1,m_2)$ runs over monomial pairs such that $m_1$ is of degree
$j-1$ and $m_2$ is of degree $d-j$. Here, $L_{m_1,m_2}$ denotes the
linear form in the variables $x_1,x_2,\ldots,x_n$ occurring in
$j^{th}$ position when we collect together all monomials of the form
$m_1x_im_2$, $1\le i\le n$. Notice that
\[
A_j(f) = \sum_{m_1,m_2}A_j(f_{m_1,m_2}) = \sum_{m_1,m_2}m_1A(L_{m_1,m_2}) m_2.
\]

Since $f\ne 0$, for some pair $(m_1,m_2)$ we have $f_{m_1,m_2}\ne
0$. In particular, $L_{m_1,m_2}\ne 0$. Since $A$ is invertible, it
follow that $m_1 A(L_{m_1,m_2}) m_2\ne 0$. Therefore, $A_j(f)\ne 0$.
\end{proof}

Given any noncommutative polynomial $f(x_1,\ldots,x_n)\in\F\angle X$
of degree $d$, we can rename the variable $x_i : 1\leq i\leq n$
appearing in the position $j\in [d]$ (from the left), by a new
variable $x_{ij}$ and obtain the polynomial $g$. We say that $g$ is
the set-multilinear polynomial obtained from $f$. One can view the
polynomial $g$ as a commutative polynomial over the variables $x_{ij}
: 1\leq i\leq n, 1\leq j\leq d$. We state a simple fact.
 
\begin{claim}\label{set-mult}
Let $f(x_1,\ldots,x_n)\in\F\angle X$ be any noncommutative polynomial
of degree $d$. For $1\leq i\leq n$, replace the variable $x_i$
appearing in the position $1\leq j\leq d$ by a new variable
$x_{ij}$. Let the new polynomial be $g(x_{11}, \ldots,x_{1d},
\ldots,x_{n1},\ldots,x_{nd})$. Then $f=0$ if and only if $g=0$.
\end{claim}

\begin{proof}
The proof follows simply from the observation that the monomials in
$f$ and $g$ are in one-one correspondence.
\end{proof}

\section{A deterministic PIT for $+$-regular circuits}\label{pit-regular}

In this section we consider noncommutative \emph{$+$-regular} circuits
defined below. These circuits can compute polynomials of exponential
degree and a double-exponential number of monomials. However,
exploiting their structure we can give a white-box deterministic
polynomial time identity test for $+$-regular circuits that proves
Theorem \ref{regular-thm}.

\begin{definition}
A noncommutative circuit $C$, computing a polynomial in $\F\angle X$
where $X=\{x_1,x_2,\ldots,x_n\}$, is \emph{$+$-regular} if it satisfies the following properties:
\begin{itemize}
\item the circuit is homogeneous.
\item The $+$ gates in the circuit are partitioned into layers such
  that if $g_1$ and $g_2$ are $+$ gates in a layer then there is no
  directed path in the circuit between $g_1$ and $g_2$.
\item The output gate is a $+$ gate in a separate layer.
\item all $+$ gates in a layer are of the same syntactic degree.
\item every input-to-output path in the circuit goes through exactly
  one $+$ gate in each layer.
\end{itemize}
\end{definition}

A simple case of $+$-regular circuits are homogeneous
$\Sigma\Pi^*\Sigma$ circuits which are defined as follows.

\begin{definition}
  A noncommutative arithmetic circuit $C$ is called a \emph{homogeneous}
  $\Sigma\Pi^*\Sigma$ circuit if it satisfies following properties:
  \begin{itemize}
\item The output gate is a $\Sigma$ gate.
\item All inputs to the output gate are $\Pi$ gates
of the same syntactic degree.
\item In the circuit, every input to output path goes through a
  $\Sigma$ gate (which computes a homogeneous linear form
  $\sum_{i=1}^n\alpha_{ij}x_i$ in the input variables
  $x_1,x_2\ldots,x_n$) followed by one or more $\Pi$ gates and ends at
  the output $\Sigma$ gate.
\end{itemize}
Likewise, we can define $\Pi^*\Sigma$ circuits.
\end{definition}

\begin{remark}
We note that, in the commutative setting, regular formulas are considered by Kayal et al. \cite{KSS14}. However, their model of regular formulas is restricted than our model of $+$-regular circuits. 
\end{remark}

The following theorem is crucial to our PIT for $+$-regular circuits.

\begin{theorem}\label{subst}
Let $T(z_1,z_2,\ldots,z_s)$ be a noncommutative homogeneous degree-$d$
polynomial over a field $\F$ in noncommuting variables
$z_1,z_2,\ldots,z_s$. Let $R_1,R_2,\ldots,R_s$ be noncommutative
homogeneous degree $d'$ polynomials in variables $x_1,x_2,\ldots,x_n$
over $\F$ such that $\{R_1,R_2,\ldots,R_k\}$ is a maximal linearly
independent subset of $\{R_1,R_2,\ldots,R_s\}$ over $\F$, where
\[
R_j=\sum_{i=1}^k\alpha_{ji}R_i,~~k+1\le j\le s,~~\alpha_{ji}\in\F.
\]
For fresh noncommuting variables $y_1,y_2,\ldots,y_k$ define
linear forms
\[
\ell_j=\sum_{i=1}^k\alpha_{ji}y_i,~~k+1\le j\le s.
\]
Then $T(R_1,R_2,\ldots,R_s)\equiv 0$ if and only if
$T(y_1,y_2,\ldots,y_k,\ell_{k+1},\ldots,\ell_s)\equiv 0$.
\end{theorem}

\begin{proof}
The reverse implication is immediate. For, suppose
$T(y_1,y_2,\ldots,y_k,\ell_{k+1},\ldots,\ell_s)\equiv 0$.  Then, by
substituting $R_i$ for $y_i, 1\le i\le k$ we obtain
$T(R_1,R_2,\ldots,R_s)\equiv 0$.

We now show the forward implication. Suppose
$T(R_1,R_2,\ldots,R_s)\equiv 0$. As $R_1,R_2,\ldots,R_k$ are linearly
independent over $\F$ we can find degree-$d'$ monomials
$m_1,m_2,\ldots,m_k$ such that the $k\times k$ matrix $B$ of their
coefficients is of full rank. More precisely, if $\beta_{ji}$ is
the coefficient of $m_i$ in $R_j$ then the matrix
\[
B=\left(\beta_{ji}\right)_{1\le j,i\le k}
\]
is full rank.

Define polynomials 

\begin{align}
R'_j &=\sum_{i=1}^k\beta_{ji}m_i, 1\le j\le k\\
R'_j &=\sum_{i=1}^k\alpha_{ji}R'_i, k+1\le j\le s.
\end{align}

Notice that $T(R_1,R_2,\ldots,R_s)\equiv 0$ implies
$T(R'_1,R'_2,\ldots,R'_s)\equiv 0$. This is because every nonzero
monomial occurring in $T(R'_1,R'_2,\ldots,R'_s)$ precisely consists of
all monomials from the set $\{m_1,m_2,\ldots,m_k\}^d$ occurring in
$T(R_1,R_2,\ldots,R_s)$ (with the same coefficient).

Replacing $m_i$ by variable $y_i, 1\le i\le k$ transforms
each $R'_j$ to linear forms
\[
\ell'_j= \sum_{i=1}^k\beta_{ji}y_i,~\textrm{ for } 1\le j\le k,
\]
and 
\[
\ell'_j= \sum_{i=1}^k\alpha_{ji}\sum_{q=1}^k\beta_{iq}y_q, \textrm{
  for } k+1\le j\le s.
\]

Note that the coefficient of any monomial $y_{i_1}y_{i_2}\ldots
y_{i_d}$ in
$T(\ell'_1,\ell'_2,\ldots,\ell'_k,\ell'_{k+1},\ldots,\ell'_s)$ is same as the coefficient of the corresponding monomial
$m_{i_1}m_{i_2}\ldots m_{i_d}$ in $T(R'_1,R'_2,\ldots,R'_s)$ which is
zero. Hence
$T(\ell'_1,\ell'_2,\ldots,\ell'_k,\ell'_{k+1},\ldots,\ell'_s)\equiv
0$. Now, since $B$ is invertible, we can apply the linear map $B^{-1}$
to each of the $d$ positions in the polynomial
$T(\ell'_1,\ell'_2,\ldots,\ell'_k,\ell'_{k+1},\ldots,\ell'_s)$ and
obtain $T(y_1,y_2,\ldots,y_k,\ell_{k+1},\ldots,\ell_s)$, which must be
identically zero by Proposition~\ref{invert}. This completes the proof.
\end{proof}

Now, suppose $C$ is a $+$-regular circuit of size $s$ of syntactic
degree $D$ computing a polynomial in $\F\angle X$, where
$X=\{x_1,x_2,\ldots,x_n\}$. Suppose there are $d$ layers of $+$-gates
in $C$, where we number the $+$-gate layers from bottom upward. Thus,
the $+$-gates in layer $1$ compute homogeneous linear forms in
$X$. Let $g_1,g_2,\ldots,g_m$ be the inputs to the layer $2$
$+$-gates. In other words, $g_1,g_2,\ldots,g_m$ are the output of the
$\times$ gates just below the layer $2$ $+$-gates.  Let $C'$ be the
circuit obtained from $C$ by deleting all gates below
$g_1,g_2,\ldots,g_m$, and replacing $g_1,g_2,\ldots,g_m$ by input
variables $y_1,y_2,\ldots,y_m$, respectively. Let
$T(y_1,y_2,\ldots,y_m)$ be the homogeneous degree $D'$ polynomial computed by $C'$. In the circuit $C$ suppose
$P_1,P_2,\ldots,P_m$ are the polynomials computed by the gates
$g_1,g_2,\ldots,g_m$, respectively. As $C$ is homogeneous, each $P_i$
is homogeneous of syntactic degree $D/D'$ (which means either $P_i$ is
identically zero or homogeneous degree $D''=D/D'$).

Notice that we can apply Theorem~\ref{subst} to the polynomials $T$
and $P_1,P_2,\ldots,P_m$, and immediately obtain the following.

\begin{lemma}\label{regpit}
Suppose, without loss of generality, that $P_1,P_2\ldots,P_t$ is a
maximal $\F$-linearly independent subset of $P_1,P_2,\ldots,P_m$,
and 
\[
P_j=\sum_{i=1}^t \alpha_{ji}P_i,~~t+1\le j\le m.
\]
Then $T(P_1,P_2,\ldots,P_m)\equiv 0$ if and only if 
$T(y_1,y_2,\ldots,y_t,\sum_{i=1}^t \alpha_{t+1 i} y_i,\ldots,\sum_{i=1}^t 
\alpha_{m i} y_i)\equiv 0$.

I.e.\ the circuit $C$ is identically zero if and only if the
circuit $C'(y_1,y_2,\ldots,y_t,\sum_{i=1}^t \alpha_{t+1 i} y_i,\ldots,\sum_{i=1}^t 
\alpha_{m i} y_i)\equiv 0$.
\end{lemma}

Clearly, Lemma~\ref{regpit} will yield a deterministic polynomial-time
identity test for regular circuits, if we can solve the following
problem in deterministic polynomial time:

Given a list of noncommutative polynomials
$P_1,P_2,\ldots,P_m\in\F\angle X$, where each $P_i$ is given by a
$\Pi^*\Sigma$ circuit, find a maximal linearly independent subset $A$
of the polynomials $P_i, 1\le i\le m$ and express the others as linear
combinations of the $P_i$ in $A$.

The PIT for $+$-regular circuits would follow because we can
repeat the same argument as above with $C'$. Finally, we will be left
with verifying if the sum of linear forms (in at most $s$ variables,
say $z_i, 1\le i\le s$) vanishes.

\subsection{Linear independence testing of $\Pi^*\Sigma$ circuits}

In this subsection we solve the above mentioned linear independence
testing problem. Namely, we prove the following theorem.

\begin{theorem}
 Given as input $\Pi^*\Sigma$ circuits computing noncommutative
 polynomials $P_1,P_2,\ldots,P_m\in\F\angle X$, there is a
 deterministic polynomial-time algorithm that will find a maximal
 linearly independent subset $A$ of the polynomials $P_i, 1\le i\le
 m$, and also express the others as $\F$-linear combinations of the
 $P_i$ in $A$.
\end{theorem}

\begin{proof}
Let $L_1,L_2,\ldots,L_t$ be the set of all linear forms (in variables
$x_1,x_2,\ldots,x_n$) defined by the bottom $\Sigma$ layers of the
given $\Pi^*\Sigma$ circuits computing polynomials $P_i, 1\le i\le m$.

Without loss of generality, let $L_1,L_2,\ldots,L_r$ be a maximal set
of linear forms among $L_1,L_2,\ldots,L_t$ that are not scalar
multiples of each other. Thus, for each $L_i, i>r$, there is some
$L_j, j\le r$ such that $L_i$ is a scalar multiple of
$L_j$. Therefore, we can express each $P_i$ as a product of linear
forms from $L_1,\ldots,L_r$, upto a scalar multiple:
\[
P_i=\alpha_i L_{i1}L_{i2}\ldots L_{iD},~~1\le i\le m
\]

Corresponding to the linear forms $L_1,L_2,\ldots,L_r$ define an
alphabet $\{a_1,a_2,\ldots,a_r\}$ of $r$ letters, where $a_i$ stands
for $L_i, 1\le i\le r$. 

Let $s$ be the bound on the sizes of the given $\Pi^*\Sigma$ circuits
computing polynomials $P_i, 1\le i\le m$.

For each $i$, we can transform the $\Pi^*\Sigma$ circuit computing
$P_i$ into a multiplicative circuit $C_i$ of size $s$ computing a word
$w_i$ of length $D$ in $\{a_1,a_2,\ldots,a_r\}^D$ as follows: replace
linear form $L_j$ in the $\Pi^*\Sigma$ circuit by letter $a_k$ if
$L_j$ is a scalar multiple of $L_k$.

The following claim is immediate.

\begin{claim}
Polynomials $P_i$ and $P_j$ are scalar multiples of each
other if and only if $w_i=w_j$.
\end{claim}

At this point we recall the following results which are implicit in 
\cite{Loh15, Pla94, MSU97} about words over
a finite alphabet, where the words are given as input by
multiplicative circuits (where multiplication is concatenation of
words).

\begin{itemize}
\item There is a deterministic polynomial time algorithm that
  takes as input two multiplicative circuits $C_i$ and $C_j$ over a
  finite alphabet and tests if the words computed by them are
  identical. If not the algorithm returns the leftmost index $k$
where the two words differ.

\item Given a word $w$ by a multiplicative circuit $C$ over some
  finite alphabet, the following tasks can be done in deterministic
  polynomial time: computing the length $|w|$ of $w$, given index $k$
  computing the $k^{th}$ letter $w[k]$, circuits $C'$ and $C''$ that
  compute the prefix $w[1\ldots k]$ and $w[k+1\ldots |w|]$ determined
  by any given position $k$, circuit $C_{k,k'}$ for the subword
  $w[k\ldots k']$ for given positions $k$ and $k'$. In particular,
  this implies that the circuit $C_{k,k'}$ is of size polynomial in
  the sizes of $C, k$ and $k'$. The parameters $k,k'$ are given in binary.
\end{itemize}

Thus, given $C_i$ and $C_j$ corresponding to polynomials $P_i$ and
$P_j$, we can find if $P_i$ and $P_j$ are scalar multiples of each
other in deterministic polynomial time.

Without loss of generality, let $P_1,P_2,\ldots,P_\ell$ be the
polynomials that are \emph{not} scalar multiples of each other.  Our
aim is to determine a maximal linearly independent subset $A$ of these
polynomials, and express each of the remaining polynomials as a linear
combination of polynomials in $A$.

Our algorithm will require a rank bound due to Saxena and Seshadri
\cite{SS13}. We recall their result first.  Consider a
$\Sigma\Pi\Sigma$ arithmetic circuit $C'$, where the top $\Sigma$ gate
has fanin $k$, all $\Pi$ gates are of fanin $D$, and each $\Pi$ gate
computes a product $Q_i, 1\le i\le k$ of homogeneous $\F$-linear forms
in commuting variables $y_1,y_2,\ldots,y_n$. Circuit $C'$ is said to
be \emph{simple} if the $gcd(Q_1,Q_2,\ldots,Q_k)=1$. Circuit $C'$ is
said to be \emph{minimal} if for any proper subset $S\subset [k]$ the
sum $\sum_{i\in S} Q_i\ne 0$.

If the polynomial computed by a simple and minimal circuit $C'$ is
identically zero then it is shown in \cite{SS13} that the rank of the
set of all $\F$-linear forms occurring at the bottom $\Sigma$ layer of
the circuit is bounded by $O(k^2\log D)$.

In order to apply this rank bound in our setting, we make the
polynomials $P_i, 1\le i\le \ell$ set-multilinear in variables
$\{x_{im}\mid 1\le i\le n, 1\le m\le D\}$ as follows: corresponding to
each linear form $L_j, 1\le j\le r$ we define linear forms $L'_{jm},
1\le m\le D$, where $L'_{jm}$ is obtained from $L_j$ by replacing
variable $x_i$ with variable $x_{im}$, for $1\le i\le n$ and $1\le
m\le D$. Likewise, we obtain the set-multilinear polynomial
$\hat{P}_i$ from $P_i$ by replacing the $m^{th}$ linear form, say
$L_j$, with $L'_{jm}$, for $1\le m\le D$.\footnote{The conversion to
  the set-multilinear polynomial is only for the sake of analysis and
  not for the actual algorithm.} The following claim is immediate.

\begin{claim}
A linear form $L'$ divides the gcd of a subset $S$ of the polynomials
$\hat{P}_i, 1\le i\le \ell$, if and only if $L'=L'_{jm}$ for some $j$
and $m$, and $L_j$ occurs in the $m^{th}$ position of each product
$P_i\in S$.
\end{claim}

The next claim includes the main step of the algorithm.

\begin{claim}
  For $i\le \ell$, we can test in deterministic polynomial time if
  $P_i$ can be expressed as an $\F$-linear combination of
  $P_1,P_2,\ldots,P_{i-1}$.
\end{claim}

\noindent\textit{Proof of Claim.}~~ Suppose $P_i$ is expressible as
an $\F$-linear combination of $P_1,P_2,\ldots,P_{i-1}$. Let $S\subseteq [i-1]$
be a minimal subset such that we can write

\[
P_i=\sum_{j\in S}\gamma_j P_j,~~\gamma_j\ne 0~~\textrm{ for all }j.
\]

Now, consider the set-multilinear circuit $C'$ defined by the sum of
products
\[
\hat{P}_i - \sum_{j\in S}\gamma_j \hat{P}_j.
\]

By minimality of subset $S$, circuit $C'$ is minimal. Suppose for some
$j\in S$, $P_i$ and $P_j$ disagree on $\rho$ positions. I.e.\ for $\rho$ positions $m$,
the linear forms occurring in the $m^{th}$ position in $P_i$ and $P_j$
are different. Let the gcd of the polynomials in the set
$\{\hat{P}_j\mid j\in S\} \cup \{\hat{P}_i\}$ be $P$, and let
$\deg(P)=\delta$. By the previous claim it follows that
\[
\delta\le D-\rho.
\]
Define polynomials $Q_i=\hat{P}_i/P$ and $Q_j=\hat{P}_j/P, j\in
S$. Notice that each $Q_i$ is a product of $D-\delta$ linear forms.
Furthermore, $\hat{P}_i - \sum_{j\in S}\gamma_j \hat{P}_j=P(Q_i -
\sum_{j\in S}\gamma_j Q_j)$. Consider the simple and minimal circuit
$C''$ defined by the sum
\[
Q_i - \sum_{j\in S}\gamma_j Q_j.
\]
Clearly, $C'$ is zero iff $C''$ is zero.  Since $C''$ is a
set-multilinear circuit, the rank of the set of all $\F$-linear forms
is at least $D-\delta$.  If $C''\equiv 0$ then by the \cite{SS13} rank
bound we have $\rho\le D-\delta\le O(\ell^2\log (D-\delta))$.  Hence,
it follows from the inequality $\log_2 x\leq x^{1/2}$ that $\rho\leq
c\cdot \ell^4$, for some constant $c>0$. Thus, for each $j\in S$,
$P_i$ and $P_j$ can disagree on at most $c\cdot\ell^4$ positions.

Therefore, the candidate polynomials $\hat{P}_j, j\le i-1$ in the
linear combination for expressing $\hat{P}_i$ are from only those
$P_j$ that disagree with $P_i$ in at most $c\cdot\ell^4$ positions. Using
the algorithms from \cite{Loh15, Pla94, MSU97} stated above, we can use the
multiplicative circuits $C_i$ and $C_j$ and determine if there are at
most $c\cdot\ell^4$ positions where the corresponding words differ. We can
also compute the at most $c\cdot\ell^4$ many indices where the words differ.
Let $S'\subseteq [i-1]$ be the set of all such indices $j$. Our goal
is to efficiently determine if $\hat{P}_i$ is an $\F$-linear
combination of the $\hat{P}_j, j\in S'$.

Let $T\subset [D]$ be the set of all positions where $P_i$ differs
from some $P_j, j\in S'$. Then $|T|\le c\cdot\ell^5$, and the polynomial
$P_i$ and all $P_j, j\in S'$ have identical linear forms in the
remaining $[D]\setminus T$ positions. Using the claim \ref{set-mult},
it is easy to see that for determining linear dependence, we can drop
the linear forms occurring in the positions in $[D]\setminus T$. Thus,
we can replace $P_i$ and each $P_j, j\in S'$ with polynomials $P'_i$
and $P'_j, j\in S'$ obtained by retaining only those linear forms
occurring in positions in $T$. \footnote{Notice that using Claim
  \ref{set-mult}, $P_i=\sum_{j\in S'} \gamma_j P_j \Leftrightarrow
  \hat{P}_i=\sum_{j\in S'} \gamma_j \hat{P}_j \Leftrightarrow Q_i =
  \sum_{j\in S'}\gamma_j Q_j \Leftrightarrow P'_i = \sum_{j\in S'}
  \gamma_j P'_j$.} We can determine these linear forms for each $P_i$
from the multiplicative circuit $C_i$ in deterministic polynomial time
using the results in \cite{Loh15, Pla94}.

Clearly each $P'_j, j\in S'$ as well as $P'_i$ is computable by a
$\Sigma\Pi\Sigma$ noncommutative circuit of size at most $O(n\ell^5)$. In
particular, these polynomials are all computable by $\poly(\ell,n)$
size noncommutative algebraic branching programs (noncommutative
ABP). Now, we will apply the main idea from the Raz-Shpilka
deterministic polynomial identity test \cite{raz05PIT} to determine if
$P'_i$ is a linear combination of the $P'_j, j\in S'$.

We explain concisely how to adapt the Raz-Shpilka algorithm. Let $B_i$
and $B_j, j\in S'$ be the ABPs computing $P'_i$ and $P'_j, j\in S'$,
respectively. Following \cite{raz05PIT} we process all the ABPs
simultaneously, layer by layer. At the $q^{th}$ layer, we maintain a
list of degree-$q$ monomials $m_{1q},m_{2q},\ldots,m_{pq}$, along with
their coefficient matrix $C{(q)}$: The $j^{th}$ columns of this matrix
gives the vector of coefficients of the monomials
$m_{1q},m_{2q},\ldots,m_{pq}$ in the polynomial computed in the
$j^{th}$ in layer $q$. Furthermore, for any other degree $q$ monomial
$m$, the coefficient vector of its coefficients at the nodes in layer
$q$ is a linear combination of the rows of $C{(q)}$. Given this data
for the $q^{th}$ layer, it is shown in \cite{raz05PIT} how to efficiently
compute the monomials and coefficient matrix for layer
$q+1$. Continuing thus, when we reach the last layer containing the
output gates of $B_i$ and $B_j, j\in S'$, we will have monomials
$m_1,m_2,\ldots,m_{\ell'}$ and corresponding ${\ell'}\times \ell$
coefficient matrix $C$ which has complete information about all linear
dependencies between the polynomials $P'_i$ and $P'_j,j\in S'$. In
particular, $P'_i=\sum_{j\in S'}\beta_jP'_j$ if and only if
$C_1=\sum_{j\in S'}\beta_j C_j$, where $C_1$ is the column of
coefficients in $P'_i$ and $C_j$ are the columns corresponding to the
$P'_j, j\in S'$, which can be determined efficiently using Gaussian
elimination. This completes the proof of this claim. \qed

To conclude the overall proof we note that the above claim can be
applied to determine the leftmost maximal linearly independent subset
$A$ of the input polynomials $P_1,\ldots, P_m$ and also express the
others as linear combinations of polynomials in $A$.
\end{proof}

\section{Black-box randomized PIT for homogeneous $\Sigma\Pi^*\Sigma$}\label{black-box-depth-three}

As shown in the previous section, we can test if a given homogeneous
$\Sigma\Pi^*\Sigma$ circuit (white-box) is identically zero in
deterministic polynomial time (as $\Sigma\Pi^*\Sigma$ circuits are
$+$-regular).

However, suppose we have only black-box access to a
$\Sigma\Pi^*\Sigma$ circuit $C$ computing a polynomial in $\F\angle
X$. I.e.\ we can evaluate $C$ on square matrices $M_i$
substituted for $x_i, 1\le i\le n$, where the cost of an evaluation is
the dimension of the $M_i$. Then it is not clear how to apply the
observations of the previous section. Specifically, $C$ may compute a
nonzero exponential degree noncommutative polynomial, but it is not
clear if we can test that by evaluating $C$ on matrices of polynomial
dimension. Also, the black-box PIT result of \cite{AMR16} cannot
be applied here since $C$ can compute polynomials of
double-exponential sparsity.

Nevertheless, we show in this section that if $C$ is an $s$-sum
$P_1+P_2+\cdots + P_s$ of $D$-products of linear forms in variables
$X$. I.e.\
\[
P_i=L_{i1}L_{i2}\ldots L_{iD},
\]
where $D$ is exponentially large then we can do black-box PIT for $C$
by evaluating it on random $O(s)\times O(s)$ matrices with entries
from $\F$ or a suitably large extension of $\F$. 
It also clearly follows from our argument that we do not need the top $+$ gate to 
be homogeneous. We only need the polynomials $P_i$ to be product of homogeneous linear 
forms. But for notational simplicity we continue writing each $P_i$ as a product of $D$ linear 
forms. 
The proof of this
claim is based on the notion of projected polynomials defined below which also shows that the 
homogeneous parts of different degrees can not participate in cancellation of terms.

\subsection{Projected Polynomials}\label{bucketing}

\begin{definition}
Let $P\in\F\angle X$ be a homogeneous degree-$D$ polynomial.  For an
index set $I\subseteq [D]$ the \emph{$I$-projection} of polynomial $P$
is the polynomial $P_I$ which is defined by letting all variables
occurring in positions indexed by the set $I$ as noncommuting. In all
other positions we make the variables commuting, by renaming $x_i$ by
the commuting variable $z_i$ for $1\leq i\leq n$. Thus, the
$I$-projected polynomial $P_I$ is in $\F[Z]\angle X$, and the
(noncommutative) degree of $P_I$ is just $I$.
\end{definition}

\begin{lemma}\label{isolating-set}
Let $P_1, P_2, \ldots, P_s \in\F\angle X$ each be a product of $D$
homogeneous linear forms
\[
P_i = L_{i,1} L_{i,2} \ldots L_{i,D},
\]
where $\{L_{i,j} : 1\leq i\leq s, 1\leq j\leq D\}$ are linear forms in
$\F\angle X$. Then there exists a subset $I\subseteq [D]$ of size at
most $s-1$ such that for any nonzero scalars
$\beta_1,\beta_2,\ldots,\beta_s\in\F\setminus\{0\}$ we have
\[
\sum_{i=1}^s \beta_i P_i=0 \textrm { iff } \sum_{i=1}^s \beta_i P_{i,I}=0,
\]
where $P_{i,I}$ is the $I$-projection of the polynomial $P_i$. 
\end{lemma}

\begin{proof}
The proof is by induction on $s$. The lemma clearly holds for
$s=1$. By induction hypothesis we assume that an index set of size at
most $s-2$ exists for a set of at most $s-1$ polynomials, each of
which is a product of $D$ homogeneous linear forms. The forward
implication is obvious, because making variables commuting can only
facilitate cancellations. We prove the reverse implication.

Suppose that $\sum_{i=1}^s \beta_i P_i \neq 0$ for nonzero $\beta_i,
1\le i\le s$. Let $j_0 \in [D]$ be the least index such that
$\rank\{L_{1,j_0}, \ldots, L_{s,j_0}\} > 1$. If no such index exists
then the $P_i$ are all scalar multiples of each other in which case
$\sum_{i=1}^s \beta_i P_i$ is just $\alpha P_1$ which is zero if and
only if $\alpha P_{1,I}$ is zero, and the implication clearly holds.

We can assume, by renumbering the polynomials that $\{L_{1,j_0},
\dots, L_{t,j_0}\}$ is a maximal linearly independent set in
$\{L_{1,j_0}, \ldots, L_{s,j_0}\}$, where $t>1$.

Then, 
\[
P_i = c_i P L_{i,j_0} L_{i, j_0+1}\ldots L_{i,D} : 1\leq i\leq t
\]
\[
P_i = c_i P \big( \sum_{k=1}^t \gamma_k^{(i)} L_{k,j_0}\big)
L_{i,j_0+1}\ldots L_{i,D} : t+1\leq i\leq s,
\]
where $\{c_i \in \F : 1\leq i\leq s\}$, $\{\gamma_k^{(i)} \in \F :
1\leq k\leq t, t+1\leq i\leq D\}$, and $P\in\F\angle X$ is a product
of homogeneous linear forms (or a scalar). For $1\leq i\leq s$, let
\[
P'_i = c_i \prod_{j=j_0 + 1}^D L_{i,j}.
\]

We can then write
\[
\sum_{i=1}^s \beta_i P_i = P\big(\sum_{i=1}^t  \beta_i L_{i,j_0} P'_i\big) +  
P\big(\sum_{i=t+1}^s  \beta_i L_{i,j_0} P'_i\big). 
\]

Note that $P\big(\sum_{i=t+1}^s  \beta_i L_{i,j_0} P'_i\big)=P\big( \sum_{i=t+1}^s \beta_i\big( \sum_{k=1}^t \gamma^{(i)}_k L_{k,j_0}\big) P'_i\big)$. Now by rearranging terms, we get the following. 

\[
\sum_{i=1}^s \beta_i P_i = P \big( \sum_{k=1}^t L_{k,j_0} P''_k\big)
\] 
where $P''_k = \beta_k P'_k + \beta_{t+1} \gamma^{(t+1)}_k P'_{t+1} + \ldots + \beta_s \gamma^{(s)}_k P'_s$ for $1\leq k\leq t$. 


Now, $L_{k,j_0}, 1\le k\le t$ are linearly independent. Applying
Proposition~\ref{invert}, consider any invertible linear map $A_{j_0}$
applied to position $j_0$ of the polynomial $\sum_{i=1}^t \beta_i P_i$
which maps $A_{j_0}: L_{k,j_0}\mapsto x_k, 1\le k\le t$. Then we
have
\[
A_{j_0}(\sum_{i=1}^s \beta_i P_i) = P\big(\sum_{k=1}^t x_k P''_k\big),
\]
and $A_{j_0}(\sum_{i=1}^t \beta_i P_i)\ne 0$. Thus, not all $P''_k, 1\le
k\le t$ are zero. Assume that $P''_1\neq 0$. Note that $P''_1$ is sum of
at most $s-1$ polynomials, each of which is a product of homogeneous
linear forms. Hence, by induction hypothesis, there is an index set
$I'\subseteq \{j_0+1, \ldots,D\}$ of size at most $s-2$ such that
\[
P''_{1,I'} = \big(\beta_1 P'_{1,I'} + \beta_{t+1}\gamma^{(t+1)}_1 P'_{t+1,I'} + \ldots +
\beta_s\gamma_1^{(s)} P'_{s,I'}\big)\neq 0\footnote{If any $\gamma^{(j)}_1$ is zero, we just work with a smaller sum.}. 
\]

Let $I=I'\cup\{j_0\}$.  Now consider the polynomial $\sum_{i=1}^s \beta_i P_{i,I}$, which we want to prove to be nonzero. Instead, we prove that $A_{j_0} \big(\sum_{i=1}^s \beta_i P_{i,I}\big)$ is nonzero. 
Notice that 
\[
 A_{j_0} \big(\sum_{i=1}^s \beta_i P_{i,I}\big)=\hat{P}\big(\sum_{k=1}^t x_k P''_{k,I'}\big),
\]
where $\hat{P}$ is the commutative polynomial obtained by replacing
$x_i$ by $z_i$ in $P$. Since $\hat{P}$ is a product of linear forms it
remains nonzero. Furthermore, the sum can be zero if and only if each $P''_{k,I'}$ is
zero. However, $P''_{1,I'}$ is nonzero. This
completes the proof.
\end{proof}


\subsection{The black-box identity test}\label{depth3pit}

We now describe a \emph{black-box} randomized polynomial time identity
testing algorithm for depth three regular circuits. Let
$C=\sum_{i=1}^s R_i$ be a polynomial in $\F\angle X$ given as
black-box, where each $R_i$ is a product of $D$ homogeneous linear
forms. By Lemma \ref{isolating-set} there is a set $I\subseteq [D]$ of
size at most $s-1$ such that $C=\sum_{i=1}^s R_i=0$ if and only if
$\tilde{C}=\sum_{i=1}^s R_{i,I}=0$. Similar to
the result in \cite{AMR16}, we will use a small size nondeterministic
automaton to guess this subset $I$ of locations, and substitute
suitable commuting variables at all locations in $[D]\setminus I$. It
will turn out that the transition matrices for each variable $x_i$
corresponding to this automaton will give us the desired black-box
substitution.

Let $|I|=k\leq s-1$. Consider the following nondeterministic finite automaton
$A$ whose transition diagram we depict for $x_i : 1\leq i\leq n$ in
Figure \ref{fig1}. For locations in $[D]\setminus I$, the automaton
uses the block variables $Z=\{z_i : 1\leq i\leq n\},  ~\xi= \{\xi_i :
1\leq i\leq k+1\}$ which are commuting variables. For each index
location $j\in I$ the automaton substitutes $x_i$ by $x_{ij}, 1\le
i\le n$, where the index variables $Z'=\{x_{ij} : 1\leq i\leq n, 1\le
j\le k\}$ are also commuting variables. 

\begin{remark}
Notice that in Lemma~\ref{isolating-set}, the variables occurring in
positions in $I$ were left as noncommuting. However, the automaton we
construct replaces $x_i$ in position $j\in I$ by commuting variable
$x_{ij}$. This transformation for homogeneous polynomials is known to
preserves identities by Claim \ref{set-mult}.
\end{remark}

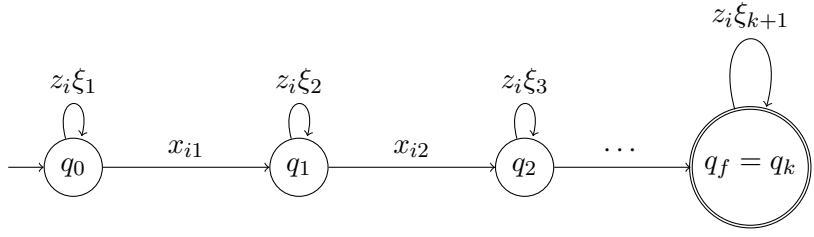
\begin{figure}
\begin{center}
\begin{tikzpicture}
\node(pseudo) at (-1,0){}; \node(0) at (0,0)[shape=circle,draw]
     {$q_0$}; \node(1) at (3,0)[shape=circle,draw] {$q_1$}; \node(2)
     at (6,0)[shape=circle,draw] {$q_2$}; \node(3) at
     (9,0)[shape=circle,draw,double] {$q_f=q_k$}; \path [->] (0) edge
     node [above] {$x_{i1}$} (1) (1) edge node [above] {$x_{i2}$} (2) (2)
     edge node [above] {$\cdots$} (3)
  (0)      edge [loop above]    node [above]  {$z_i\xi_1$}     ()
  (1)      edge [loop above]    node [above]  {$z_i\xi_2$}     ()
(2)      edge [loop above]    node [above]  {$z_i\xi_3$}     ()
  (3)      edge [loop above]    node [above]  {$z_i\xi_{k +1}$}   ()
  (pseudo) edge                                       (0);
\end{tikzpicture}
\caption{The transition diagram for the variable $x_i : 1\leq i\leq n$}\label{fig1}
\end{center}
\end{figure} 

Let
\[
\forall i\in [s]   :   R_i = L_{i,1} \ldots L_{i,D}.   
\]

Let $M_{x_i}$ be the matrix corresponding to variable $x_i, 1\le i\le
n$. When we do this matrix substitution to variables in $R_i$, the
$(0,k)^{th}$ entry of the resulting matrix $M_{R_i}$ is
\[
\widehat{R}_i = \sum_{(j_1,j_2,\ldots,j_k)\in [D]^k}\prod_{j=1}^{j_1 -
  1} L_{i,j}(Z) {\xi_1}^{j_1-1} L_{i,j_1}(Z') \prod_{j=j_1+1}^{j_2-1}
L_{i,j}(Z){\xi_2}^{j_2-j_1-1} L_{i,j_2}(Z')\ldots\ldots
\]

For each $i\in [s]$, the polynomial $\widehat{R}_i\in\F[Z,\xi,Z']$.
The $(0,k)^{th}$ entry of the resulting matrix $M_C$ is 
\[
\sum_{i=1}^s \widehat{R}_i = \sum_{J\in[D]^k} P_J \xi_J,
\] 
where $\xi_J = \xi_1^{j_1-1} \xi_2^{j_1 - j_2 -1}\dots \xi_k^{D-j_k}$
and $P_J=\sum_{i=1}^s {P_{i,J}}$.

By Lemma \ref{isolating-set}, we know that $P_I=\sum_{i=1}^s
P_{i,I}\neq 0$. Thus, $\sum_{i=1}^s \widehat{R}_i$ is nonzero, as the
monomials sets for different $P_J$ are disjoint (ensured by the terms
$\xi_J$). The degree of  $\sum_{i=1}^s \widehat{R}_i$ is $D$. So if $|\F|$ is more than $D$, it can not evaluate to zero on $\F$. This completes the proof of Theorem \ref{depth-3-pit}. 

Now the randomized identity testing algorithm follows by
simply random substitution for variables in the commutative polynomial
computed at the $(0,k)^{th}$ entry of the resulting matrix $M_C$. 
This completes the proof of Corollary \ref{depth-3-algo}. 


\section{Conclusion}\label{conclusion}

The main open problem is to find a randomized polynomial time identity
test for general noncommutative circuits (in the white-box model). Our
result for $+$-regular circuits is a first step towards that. Finding
an efficient randomized black-box identity testing algorithm for
$+$-regular circuits is also an interesting problem. For homogeneous
$\Sigma\Pi^*\Sigma$ circuits, we have obtained such a randomized
black-box identity test.

\end{document}